\documentclass{article}

\usepackage{epsfig,amsfonts,amsthm,amssymb,latexsym,amsmath,a4wide,tikz,array}
\usepackage{verbatim,a4wide}
\usepackage{dsfont}
\usepackage{rotating}
\usepackage{multirow}
\usepackage[utf8]{inputenc}

\theoremstyle{plain}
\newtheorem{theorem}{Theorem}
\newtheorem{corollary}{Corollary}
\newtheorem{lemma}{Lemma}
\newenvironment{problem}[2][Problem]{\begin{trivlist}
\item[\hskip \labelsep {\bfseries #1}\hskip \labelsep {\bfseries #2}]}{\end{trivlist}}

\renewcommand{\vec}[1]{\boldsymbol{#1}}

\begin{document}

\title{Scheduling  under   dynamic  speed-scaling   for  minimizing
  weighted completion time and energy consumption}

\author{Christoph Dürr\thanks{CNRS, LIP6, Université Pierre et Marie
    Curie, Paris, France.}
    \and
    {{\L}ukasz Je{\.z}}\thanks{Blavatnik School of Computer Science, Tel Aviv University, Israel, and Institute of  Computer Science, University of  Wroc{\l}aw, Poland.}
    \and
    {Oscar C. Vásquez}\thanks{LIP6 and Industrial Engineering Department, University of Santiago of
  Chile.}
  }

\maketitle

\begin{abstract}
Since a  few years there is  an increasing interest  in minimizing the energy  consumption  of  computing  systems.  However in  a  shared  computing system, users  want to optimize their experienced  quality of service, at the price of a high  energy consumption.  In this work, we address the problem of optimizing and designing mechanisms for a linear combination of weighted completion time and energy consumption on a single machine with dynamic speed-scaling.  We show that minimizing linear combination reduces to a unit speed scheduling problem under a polynomial penalty function. In the mechanism design setting, we define a cost share mechanism and studied its properties, showing that it is \textit{truthful} and the overcharging of total cost share is bounded by a constant.
\end{abstract}

\paragraph{Keywords:}
scheduling, energy management, quality of service, optimization, mechanism design.

\section{Introduction}
Humanity has entered a period when natural resources become rare. This situation triggered consciousness in responsible consumption, and many countries, companies and individuals aim to minimize their energy consumption. Minimizing energy consumption is a relatively new topic in decision theory, giving rise to new problems and research areas.

An area of increasing interest is the energy consumption minimization of computing systems with dynamic speed-scaling, allowed in modern microprocessors by technologies such as Intel SpeedStep, AMD PowerNow!, or IBM EnergyScale.  The theoretical energy consumption model has been introduced in \cite{YaoDemmersShenker:95:Energy-minimization}, and triggered the development of offline and online algorithms; see~\cite{Albers:10:Survey-energy-management} for an overview.

In these systems, minimizing energy consumption of the machines and minimizing waiting times of the users are opposed goals~\cite{IraniPruhs:05:Problems-energy-management}. Small waiting times improve the quality of service experienced by the users, which generally comes at the price of high energy consumption.

The online and offline optimization problem for minimizing flow time, while respecting a maximum energy consumption, has been studied for a single machine in \cite{AlbersFujiwara:07:Some-minimization-energy-QoS,BansalKimbrelPruhs:07:Proof-optimality-YDS-KKT,PruhsUthaisombutWoeginger:08:Minimization-unweighted-completion-time-constraint-energy,ChanLamLee:10:Some-minimization-energy-QoS} and for parallel machines in \cite{AngelBampisKacem:12:Parallel-unrelated-machine-scheduling-energy}. For the  variant  where  an  aggregation   of  energy  and  flow  time  is considered, polynomial time approximation algorithms have been presented in \cite{CarrascoIyegarStein:11:Some-minimization-energy-QoS,BansalChanKatzPruhs:12:Some-minimization-energy-QoS,MegowVerschae:13:PTAS-piecewise-penalty-function}.

In this paper, we study this problem from a decentralized and realistic perspective.  Here in order to optimize the objective function, a service operator needs specific information from the users together with the characteristics of the submitted jobs. This situation creates the need for a truthful mechanism.

Specifically, we consider a simplified computing system with a single shared machine using dynamic speed-scaling, meaning it can run at a variable continuous speed to influence the completion times of the jobs. Users submit jobs to  this system,  each  job has some  workload,  representing a number of instructions to execute and a delay penalty factor. All jobs are available from time 0 on.  During the submission of a job, only the workload is publicly known, while the players might announce false delay penalties in order to influence the game towards their interest.

The machine is controlled by an \textit{operator} who aims to minimize the sum of the total weighted completion time and of the energy consumption cost. To this end he decides both on the speed function of the machine, and the order in which jobs are to be scheduled.  The energy cost consumed by the schedule needs to be charged to the users. The individual goal of each user is to minimize the sum of waiting time and the energy cost share.  Therefore it is the charging scheme chosen by the operator that influences the players' behavior.

In a companion paper~\cite{DurrJezVasquez:13:Marginal-Game}, we study a similar game, where the players announce a strict deadline for their job, while keeping the delay penalty factor private. This way the players control the quality of service guaranteed by the operator, leaving to the operator the goal of optimizing the consumed energy under this constraint.

\section{The Model}

Consider a non-cooperative game with $n$ users and a computing system with a single shared machine using dynamic speed-scaling. Each player has a single job $i$ with a positive workload $w_i$ and a positive delay penalty $p_i$. When the player submits  his job, he announces the workload and some delay penalty $\hat p_i$.  The announced value might differ from the real value, in order to influence the game towards his advantage, while the workload has to be the true value.  The latter assumption makes sense, since it is a quantity observable by the operator, who could punish players in case they lied on the workload.

The machine is controlled by an operator, who upon reception of the jobs, has to produce a schedule.  This schedule generates some energy consumption, and the controller needs to charge this value to the players, according to some \emph{charging scheme}. This charging scheme is known in advance to all users.

A schedule is defined by an execution order $\pi$ and an execution speed.  Following  \cite{YaoDemmersShenker:95:Energy-minimization} it is assumed that every job $i$ is scheduled at constant speed, and for the purpose of simplifying notation we rather specify the execution length $\ell_i$ of every job $i$, rather than its speed $s_i$, which is $w_i/\ell_i$.  Two costs are associated with a schedule: the energy cost
\[
     E(\vec{\ell}, \vec{w}) := \sum_i \ell_i s_i^\alpha = \sum_i w_i^\alpha \ell_i^{1-\alpha}
\]
defined for some fixed physical constant $2\leq \alpha \leq 3$, and the weighted flow time, representing the quality of service delivered to the users,
\[
    F(\pi,\vec{\ell}, \vec{p}) := \sum_i p_i C_i,
\]
where $C_i$ is the completion time of job $i$, defined as $\sum_{j: \pi(j)\leq \pi(i)} \ell_j$, with $\pi(j)$ being the rank of job $j$ in the schedule.

Ideally the operator would like to minimize the total cost $E(\vec{\ell},\vec{w})+F(\pi,\vec{\ell},\vec{p})$, which we call the \emph{social cost}.  When adding the two costs, we consider the conversion of energy and completion time into monetary values, and assume for simplification that the conversion factors are hidden in the penalty factors.
When optimizing the social cost, some balance between the two components has to be found, because the energy cost is small when the execution lengths are large, while for the weighted flow time it is exactly opposite.  Several problems arise in this situation.

First, the game operator knows only the announced penalty factors $\vec{\hat p}$, so the game has to be \emph{truthful}. This means that every player optimizes his cost by announcing the true value $\hat p_i= p_i$.

Second, no polynomial algorithm for finding a schedule $\pi,\vec{\ell}$ that minimizes $E(\vec{\ell},\vec{w})+F(\pi,\vec{\ell},\vec{p})$ for arbitrary value $\alpha$ is known.  In fact, it is also not known whether this
problem is NP-hard.
However, a PTAS is known~\cite{MegowVerschae:13:PTAS-piecewise-penalty-function}.
Dominance properties for this problem have been established in \cite{DurrVasquez:13:Rules}.
We note that the solution boils down to finding the right order
of scheduling, since once $\pi$ is fixed, the optimum durations (speeds) for processing each job can be easily determined,
cf.~\cite{MegowVerschae:13:PTAS-piecewise-penalty-function} or
Section~\ref{sec:equivalence}.

The operator charges to every player $i$ some value $b_i$, depending on the submitted jobs and on the constructed schedule.  This charge has two roles. On one side it is supposed to cover the energy cost of the schedule, and on the other side it influences the players behavior, as every player $i$ wants to minimize the sum of the weighted completion time of his job plus his cost share
\(
         p_i C_i + b_i.
\)

\subsection{Desirable properties}

In summary,
we want to design a cost sharing mechanism that is
\begin{itemize}
\item \textit{truthful}, meaning that every player minimizes his penalty by  announcing his true value, i.e. $\vec{\hat p}=\vec{p}$. This implies that the strategy profile $\vec{p}$ is a pure Nash equilibrium.
\item \textit{$\beta$-budget-balanced} for some constant $\beta$,  meaning that the sum of cost shares is at least the energy cost and at most $\beta$ times this value.
\item \emph{efficient}, meaning that the social cost of the Nash equilibrium is close to the social cost optimum.
\end{itemize}

\section{Optimizing social cost}\label{sec:equivalence}

We consider the centralized optimization problem consisting in minimizing the sum of the energy consumption and of the total weighted completion times, under the assumption that the regulator knows the true penalty factors $\vec{p}$ of all players. Formally it can be stated as follows.

\begin{problem}{A} Given $n$ jobs with workloads $\vec{w}$ and penalty factors $\vec{p}$, find a schedule defined by $\pi,\vec{\ell}$ which minimizes $A_{\vec{w},\vec{p}}(\pi,\vec{\ell}) = E(\vec{\ell},\vec{w})+F(\pi,\vec{\ell},\vec{p})$.
\end{problem}

This problem is equivalent to another scheduling problem.  Here the machine runs at uniform speed, and the role of the speed is encoded in the objective function.  Note that the two values of the input $\vec{w},\vec{p}$ play opposite roles in both problems.

\begin{problem}{B} Given $n$ jobs with priority weights $\vec{w}$ and processing times $\vec{p}$, find a schedule defined by an order $\sigma$ which minimizes $B_{\vec{w},\vec{p}}(\sigma)=\sum w_j C_j^{(\alpha-1)/\alpha}$, where the completion time $C_j$ is defined as $\sum_i p_i$ over all jobs $i$ with $\sigma(i) \leq \sigma(j)$.
\end{problem}

In this section we show equivalence between the two problems.  Our proof uses the Hessian conditions. Independently, a similar reduction has been discovered in~\cite{MegowVerschae:13:PTAS-piecewise-penalty-function} using Karush-Kuhn-Tucker conditions, relating Problem B and a variant of Problem A, with the goal of minimizing total weighted completion time under a given energy budget.

\begin{theorem} \label{theo:reduction} Let $\pi$ be a job order, and denote by $\sigma$ its reverse, i.e.\ $\sigma(i)=n+1-\pi(i)$.  Let $\vec{\ell}$ be the execution length vector minimizing $A_{\vec{w},\vec{p}}(\pi,\vec{\ell})$.
Then $A_{\vec{w},\vec{p}}(\pi,\vec{\ell}) = \alpha(\alpha-1)^{(\alpha-1)/\alpha} \cdot B_{\vec{w},\vec{p}}(\sigma)$.
\end{theorem}

\begin{proof}
Fix permutations $\pi,\sigma$ and vector $\vec{\ell}$ with the required condition. Without loss of generality, we assume that $\pi(i)=i$ (and $\sigma(i)=n+1-i$).  This can always be achieved by renaming the jobs, since the problems are independent on the actual job indices.

The objective value of problem A is
\begin{equation}
\label{eq:reformulte}
A_{\vec{w},\vec{p}}(\pi,\vec{\ell})= \sum_{j=1}^n w_{j}^{\alpha}\ell_j^{1-\alpha} + \sum_{j=1}^n p_{j} \sum_{k=1}^j \ell_k.
\end{equation}

For any job $j$, the value above must be a local minimum with respect to $\ell_j$, meaning that its derivative is zero. In other words
\[
(1-\alpha) w_{j}^{\alpha} \ell_j^{-\alpha} + \sum_{k=j}^n p_{k}=0,
\]
or equivalently
\begin{align}
s_j &= w_{j} / \ell_j = \left( \frac{\sum_{k=j}^n p_{k}}{\alpha - 1} \right)^{\frac{1}{\alpha}}\enspace,\label{eq:wl}\\
     \ell_j &= w_{j} \cdot \left(\frac{(\alpha-1)}
                                   {\sum_{k=j}^n  p_{k}}
                                 \right)^{\frac 1\alpha}
                                 \enspace. \label{eq:l}
\end{align}
This condition is sufficient, since the Hessian is positive definite. In fact,
the first derivative in $\ell_j$ is independent of $\ell_i$ for any $i\neq j$, and so
the Hessian of  $A$ has zero non-diagonal terms, whereas the second
derivative of $A$ in $\ell_j$ is
\[
        w_{j}^\alpha(\alpha-1)\alpha /\ell_j^{\alpha+1},
\]
which is positive for positive $\ell_j$ and $\alpha>1$. Thus, we have that the diagonal terms of the Hessian of $A$ are positive, the Hessian is positive definite and then, $A$ is minimized by a vector $\vec{\ell}$ which sets to zero the first derivative for every $j$.

Now, we rewrite the second term of the sum in \eqref{eq:reformulte}
\begin{equation}
\label{ref:relationflowtime}
\sum_{j=1}^n p_j \sum _{k=1}^j \ell_k=\sum_{j=1}^n \ell_j \sum_{i=k}^n p_k
\end{equation}
and have,
\begin{align}
\label{eq:reformulte2}
A_{\vec{w},\vec{p}}(\pi,\vec{\ell}) = \sum_{j=1}^n w_{j}^{\alpha}\ell_j^{1-\alpha}+ \sum_{j=1}^n \ell_j \sum_{k=j}^n p_k=& \sum_{j=1}^n \ell_j \left(w^{\alpha}_{j} \ell_j^{-\alpha}+\sum_{k=j}^n p_k\right)
\end{align}
Finally, we replace \eqref{eq:wl} for $w^{\alpha}_{j} \ell_j^{-\alpha}$ and then \eqref{eq:l} for $\ell_j$ in \eqref{eq:reformulte2} yields
\begin{align*}
 A_{\vec{w},\vec{p}}(\pi,\vec{\ell})&=\sum_{j=1}^n \ell_j \left(\frac{\sum_{k=j}^n p_{k}}{\alpha - 1} + \sum_{k=j}^n p_{k} \right)
\\
 &= \frac{\alpha}{\alpha-1} \sum_{j= 1}^n \ell_j \sum_{k=j}^n p_{k}\\
 &=\frac{\alpha}{\alpha-1} \sum_{j=1}^n \left( \frac{ w_{j}^\alpha (\alpha-1) }
                                   {\sum_{k=j}^n p_{k}} \right)^{\frac 1\alpha} \sum_{k=j}^n p_{k} \\
&= \alpha (\alpha-1)^{\frac{1-\alpha}\alpha} \sum_{j=1}^n w_{j} \left(\sum_{k=j}^n p_{k} \right)^{\frac{\alpha-1}\alpha}\\
&= \alpha (\alpha-1)^{\frac{1-\alpha}\alpha} \cdot B_{\vec{w},\vec{p}}(\sigma)\enspace,
\end{align*}
since scheduling jobs with processing times $\vec{p}$ in order $n,n-1,\ldots,1$ on a uniform speed machine, results in completion time
\[
    C_j = \sum_{k=j}^n p_{k}
\]
for every job $j$.
\end{proof}

\section{Mechanism design problem}

Clearly the operator can optimize the social cost only if the players announce the true penalty factors, otherwise the operator optimizes with false values, and we have no guarantee on the outcome.  This motivates the design of a truthful mechanism.

We design the mechanism as follows.  Fix an arbitrary job order $\pi$.
Denote the workload vector and the penalty vector as declared by the players by $\vec{w}$ and
$\vec {\hat p}$ respectively.  Note that we do assume that $\vec{w}$ is the vector of true workloads
but do not assume that $\vec {\hat p}$ is the vector of true penalties.
Given $\pi$, $\vec{w}$, and $\vec {\hat p}$, the mechanism is going to schedule all jobs
in the order given by $\pi$, setting the speed for job $j$ so that the time it takes to process it
is $\ell_j$ given by equation~\eqref{eq:l}.
We stress again that the mechanism can choose any order, for example uniformly at random,
but it is crucial that the order is \emph{independent} of the players strategies, i.e.,
the penalty factors they declare.

As suggested, with fixed order $\pi$, Theorem \ref{theo:reduction} applies, providing the vector
$\vec{\ell}$ of processing times that minimizes the social cost for this particular $\pi$.
In particular, with fixed order the underlying optimization problem is solved in polynomial time,
enabling practical implementation.

Let $OPT_{\pi}(\vec{\hat p})$ denote the value $E(\vec{\ell},\vec{w})$ with $\vec{\ell}$ being the minimizer for $E(\vec{\ell},\vec{w})+F(\pi,\vec{\ell},\vec{\hat p})$, which is the energy component of the social optimum. In addition we denote by $OPT_{\pi}(\vec{\hat p}_{-i})$ the similar value when player $i$ is excluded from the game. Formally if for a vector $v$ of dimension $n$ we denote by $v_{-i}$ the $(n-1)$-dimensional vector resulting from the deletion of the $i$-th element, then $OPT_{\pi}(\vec{\hat p}_{-i})$ denotes the value $E(\vec{\ell'},\vec{w}_{-i})$
 with $\vec{\ell'}$ being the minimizer for $E(\vec{\ell'},\vec{w}_{-i})+F(\pi_{-i},\vec{\ell'},\vec{\hat p}_{-i})$.  Note that in order to simplify the notation, we dropped the $-i$ subscript on $\pi$, in the notation  $OPT_{\pi}(\vec{\hat p}_{-i})$.

The operator defines a cost sharing scheme where player $i$ pays the penalty
\[
           b_i:= \alpha \left(OPT_{\pi}(\vec{\hat p})-OPT_{\pi}(\vec{\hat p}_{-i})\right)- \hat p_i C_i,
\]
and the executing length and speed vector are the minimizers for $OPT_{\pi}(\vec{\hat p})$.

Note that lengths and speeds of job $j>i$ minimizing $OPT_{\pi}(\vec{\hat p}_{-i})$ and $OPT_{\pi}(\vec{\hat p})$ are the same, by equation~\eqref{eq:l}.

\subsection{Truthfulness}
In this section, we prove that the mechanism admits a unique Nash equilibrium, which is truthful.  Formally, we claim that for every player $i$ the strictly dominant strategy is $\hat{p_i}=p_i$.

To show the above claim, we will need the following technical lemma.
\begin{lemma}
\label{lem:prop}
Consider $\alpha>1$, a fixed order $\pi$, and an arbitrary player $i=\pi(j)$. For all $k \leq \pi(j)$, we have:
\begin{align}
\alpha \frac{\partial s_k^{\alpha} \ell_k}{\partial \hat p_i}=& \ell_k \label{eq:lemprop1}
\\
\frac{\partial \ell_k}{\partial \hat p_i} < &0.\label{eq:lemprop2}
\end{align}
\end{lemma}

\begin{proof}
Fix a permutation $\pi$, a value $\alpha>1$ and an arbitrary player $i$. Without loss of generality, we assume that $i=\pi(i)$ for all $i=1,\ldots,n$.

Note that it follows from~\eqref{eq:wl} and~\eqref{eq:l} that for every $k \leq i$, we have
\begin{equation*}
s_k^{\alpha} \ell_k = (w_k/\ell_k)^\alpha \cdot \ell_k =
	\frac{w_k}{(\alpha-1)^{1-1/\alpha}}\left(\sum_{j=k}^n \hat p_j \right)^{1-1/\alpha} \enspace.
\end{equation*}

By taking a partial derivative of the above expression in $\hat p_i$, we obtain
\begin{align*}
\frac{\partial s_k^{\alpha} \ell_k}{\partial \hat p_i}=&\frac{w_k}{(\alpha-1)^{1-1/\alpha}}\frac{\partial\left(\sum_{j=k}^n \hat p_j \right)^{1-1/\alpha}}{\partial \hat p_i}\\
                                                  =&\frac{w_k}{(\alpha-1)^{1-1/\alpha}}\frac{\alpha-1}{\alpha}\left(\sum_{j=k}^n \hat p_j \right)^{-1/\alpha}\\
                                                  =&\frac{w_k}{\alpha (\alpha-1)^{-1/\alpha} }\left(\sum_{j=k}^n \hat p_j\right)^{-1/\alpha}\\
                                                  =&\frac{w_k}{\alpha}\left(\frac{\sum_{j=k}^n \hat p_j}{\alpha-1} \right)^{-1/\alpha}\\
                                                  =&\frac{w_k}{\alpha}  \left(\frac{\alpha-1}{\sum_{j=k}^n \hat p_j}\right)^{1/\alpha}\\
                                                  =&\frac{\ell_k}{\alpha} \enspace,
\end{align*}
where the last identity follows from~\eqref{eq:l}.  This immediately implies \eqref{eq:lemprop1} by
multiplying by $\alpha$.

Finally, we note that~\eqref{eq:lemprop2} immediately follows from~\eqref{eq:l}, as for all $k \leq i$,
\begin{equation*}
\frac{\partial \ell_k}{\partial \hat p_i} = w_k (\alpha-1)^{1/\alpha}\frac{\partial\left(\sum_{j=k}^n \hat p_i \right)^{-1/\alpha}}{\partial \hat p_i}
= \frac{-w_k (\alpha-1)^{1/\alpha}}{\alpha}\left(\sum_{j=k}^n \hat p_j \right)^{-1/\alpha-1}<0 \enspace.
\end{equation*}
\end{proof}

We now show the main result.
\begin{theorem}
\label{teo:strategyproof}
The mechanism is truthful.
\end{theorem}

\begin{proof}
We need to show that every player $i$ minimizes his penalty when choosing the strategy $\hat p_i =p_i$. For the ease of notation, without loss of generality we assume $\pi(i)=i$. The total penalty of player $i$ is:
\begin{align*}
&p_i C_i + \alpha \left(OPT_{\pi}(\vec{\hat p})-OPT_{\pi}(\vec{\hat p}_{-i})\right)- \hat{p_i} C_i \\
=&(p_i - \hat p_i) \sum_{k=1}^i \ell_k + \alpha \left(OPT_{\pi}(\vec{\hat p})-OPT_{\pi}(\vec{\hat p}_{-i})\right) \enspace.
\end{align*}

By taking a partial derivative in $\hat p_i$, we get
\begin{eqnarray*}
& \left(-\sum_{k=1}^i \ell_k+(p_i-\hat p_i) \sum_{k=1}^i \frac{\partial \ell_k}{\partial \hat p_i}\right) + \alpha \left(\sum_{k=1}^i \frac{\partial s_k^{\alpha}\ell_k}{\partial\hat p_i}\right)\\
=& -\sum_{k=1}^i \ell_k+(p_i-\hat p_i) \sum_{k=1}^i \frac{\partial \ell_k}{\partial \hat p_i} + \sum_{k=1}^i \ell_k\\
=& (p_i-\hat p_i) \sum_{k=1}^i \frac{\partial \ell_k}{\partial \hat p_i} \enspace,
\end{eqnarray*}
by \eqref{eq:lemprop1}.  By Lemma~\ref{lem:prop} \eqref{eq:lemprop2},
each summand in the last expression is negative, which implies that $\hat p_i= p_i$
minimizes the player's total penalty.
\end{proof}

\subsection{Efficiency and overcharging}

In this section, we estimate the efficiency and overcharging of the total cost share when at least 2 players participate in the game.

\begin{theorem}
\label{teo:overcharging}
The sum of the cost shares is at least one and at most $\alpha+1$ times the optimum social cost.
\end{theorem}

\begin{proof}
Without loss of generality, we assume $\pi(i)=i$.  We start with the lower bound,
proving that the cost share $b_i$ of player $i$ is no smaller than the energy consumed by
this player's job.  By repeatedly using~\eqref{eq:wl} and~\eqref{eq:l}, in particular observing
that the speed and processing time of a job $k$ depends only on the parameters of jobs $i\geq k$,
i.e., for $k>i$ they are the same in $\vec{\hat p}$ and $\vec{\hat p}_{-i}$, we can express
$b_i$ as follows.
\begin{align}
b_i=&\alpha \left(OPT_{\pi}(\vec{\hat p})-OPT_{\pi}(\vec{\hat p}_{-i})\right)-\hat p_i C_i\notag\\
   =&\alpha\left(\sum_{k=1}^i \ell_k s_k^{\alpha}- \sum_{k=1}^{i-1} w_k \left(\frac{\sum_{j=k}^n \hat p_j-\hat p_i}{\alpha-1}\right)^{1-1/\alpha} \right)-\hat p_i \sum_{k=1}^i \ell_k \notag\\
   =&\alpha \ell_i s_i^{\alpha}- \ell_i \hat p_i +\sum_{k=1}^{i-1} \left(\alpha \ell_k s_k^{\alpha}- \ell_k \hat p_i - \alpha w_k \left(\frac{\sum_{j=k}^n \hat p_j-\hat p_i}{\alpha-1}\right)^{1-1/\alpha} \right)\notag\\
   =& \ell_i \left( \alpha s_i^{\alpha}- \hat p_i \right) +\sum_{k=1}^{i-1} \left(\alpha w_k \left(\frac{\sum_{j=k}^n \hat p_j}{\alpha-1}\right)^{1-1/\alpha}  - \alpha w_k \left(\frac{\sum_{j=k}^n \hat p_j-\hat p_i}{\alpha-1}\right)^{1-1/\alpha} -\ell_k \hat p_i \right)\notag\\
   =& \ell_i \left( \alpha s_i^{\alpha}- \hat p_i \right) +\sum_{k=1}^{i-1} \left(\frac{\alpha}{(\alpha-1)^{1-1/\alpha}} w_k \left(\left(\sum_{j=k}^n \hat p_j \right)^{1-1/\alpha} - \left(\sum_{j=k}^n \hat p_j-\hat p_i\right)^{1-1/\alpha}\right) -\ell_k \hat p_i \right)\notag\\
   >& \ell_i \left( \alpha s_i^{\alpha}- \hat p_i \right) +\sum_{k=1}^{i-1} \left(\frac{\alpha}{(\alpha-1)^{1-1/\alpha}} w_k \hat{p_i} \frac{\alpha-1}{\alpha}\left(\sum_{j=k}^n \hat p_j\right)^{-1/\alpha}-\ell_k \hat p_i\right) \label{eq:concavity}\\
   =& \ell_i \left( \alpha s_i^{\alpha}- \hat p_i \right) + \hat p_i \sum_{k=1}^{i-1} \left(\frac{w_k}{(\alpha-1)^{-1/\alpha}} \left(\sum_{j=k}^n \hat p_j\right)^{-1/\alpha}- \ell_k \right)\notag\\
   =& \ell_i \left( \alpha s_i^{\alpha}- \hat p_i \right) +\hat p_i \sum_{k=1}^{i-1} \left(\ell_k - \ell_k \right)\notag\\
=& \ell_i \left( \alpha s_i^{\alpha}- \hat p_i \right) \notag\\
=&\ell_i \left( s_i^{\alpha}+(\alpha-1) s_i^{\alpha}- \hat p_i \right) \notag\\
=&\ell_i \left( s_i^{\alpha}+ \sum_{k=i}^n \hat p_k- \hat p_i \right) \notag\\
=&\ell_i \left( s_i^{\alpha}+ \left( \sum_{k=i+1}^n \hat p_k\right) \right) \notag\\
>&\ell_i s_i^{\alpha} \notag \enspace,
\end{align}
where the inequality \eqref{eq:concavity} follows from the strict concavity of function $f:t \mapsto t^{1-1/\alpha}$ for $\alpha>1$, which implies
\[f(t_1)-f(t_1-a)=t_1^{1-1/\alpha}-(t_1-a)^{1-1/\alpha}>a f'(t_1)=a (1-1/\alpha) t_1^{-1/\alpha}.\]

Thus, we have that the energy consumption of user $i$ is $\ell_i s_i^{\alpha}$ and then, we have that the cost share $b_i$  is at least its cost energy consumption, which concludes the first part of the proof.

For the upper bound, we denote by $\vec{\ell'}$ the execution length vector which is the minimizer for $OPT_{\pi}(\vec{\hat p}_{-i})$.  For convenience, we denote the corresponding speed $s'_k := w_k / \ell'_k$.
\begin{align*}
OPT_{\pi}(\vec{\hat p})-OPT_{\pi}(\vec{\hat p}_{-i}) &=\sum_{k=1}^{n} (\ell_k s_k^{\alpha} - \ell'_k s'^{\alpha}_k) \\
                           &=\sum_{k=1}^{i-1} (\ell_k s_k^{\alpha} - \ell'_k s'^{\alpha}_k)+\ell_i s_i^\alpha\\
                           &=\sum_{k=1}^{i-1} \left(\ell_k \frac{\sum_{j=i}^n \hat p_j}{\alpha - 1}-\ell'_k\frac{\sum_{j=i}^n \hat p_j-\hat p_i}{\alpha - 1}\right) +\ell_i s_i^\alpha\\
                           &=\sum_{k=1}^{i-1} \left((\ell_k-\ell'_k) \frac{\sum_{j=i}^n \hat p_j-\hat  p_i }{\alpha - 1}\right) + \frac{\hat  p_i}{\alpha-1} \sum_{k=1}^{i-1} \ell_k +\ell_i s_i^\alpha\\
                           &<\frac{\hat  p_i}{\alpha-1} \sum_{k=1}^{i-1}\ell_k +\ell_i s_i^\alpha\\
                           &<\frac{\hat  p_i}{\alpha-1} \sum_{k=1}^{i}\ell_k +\ell_i s_i^\alpha.
\end{align*}
The inequality follows from $\ell_k<\ell'_k$ for every $1\leq k <i$.

We have the following bound on the cost share on player $i$,
\[
\ell_i s_i^\alpha \leq b_i \leq \alpha \left(\frac{\hat  p_i}{\alpha-1} \sum_{k=1}^{i} \ell_k +\ell_i s_i^\alpha \right) - \hat p_i  \sum_{k=1}^{i} \ell_k,
\]
which summed up over all players leads to
\begin{align*}
OPT_{\pi}(\vec{\hat p}) \leq \sum_{i=1}^n b_i \leq& \alpha OPT_{\pi}(\vec{\hat p})+ \frac{1}{\alpha-1} \sum_{i=1}^n \hat p_i \sum_{k=1}^{i}\ell_k\\
                                           = & \alpha OPT_{\pi}(\vec{\hat p}) + \frac{1}{\alpha-1} \sum_{i=1}^n \ell_i \sum_{k=i}^{n} \hat p_k\\
                                           =& \alpha OPT_{\pi}(\vec{\hat p}) + \frac{1}{\alpha-1} (\alpha -1) \sum_{i=1}^n \ell_i  s_i^\alpha\\
                                           =& \alpha OPT_{\pi}(\vec{\hat p}) + OPT_{\pi}(\vec{\hat p})\\
                                           =&(\alpha+1)OPT_{\pi}(\vec{\hat p}),
\end{align*}
concluding the proof.
\end{proof}

\begin{corollary}
\label{cor:efficient}
The social cost of Nash equilibrium is at least one and at most $(\alpha+1)$ times the optimal social cost.
\end{corollary}

\begin{proof}
Let $\pi$ and $\vec{\ell}$ and be a fixed order job and the execution length vector minimizing
\[A_{\vec{w},\vec{p}}(\pi,\vec{\ell})=E(\vec\ell,\vec{w})+F(\pi,\vec{\ell}, \vec{p}).\]
It suffices to show
\[A_{\vec{w},\vec{p}}(\pi,\vec{\ell})=\alpha E(\vec{\ell},\vec{w}),\]
or equivalent
\[(\alpha-1)E(\vec{\ell},\vec{w})=F(\pi,\vec{\ell}, \vec{p}).\]
We have
\[(\alpha-1)E(\vec{\ell},\vec{w})=(\alpha-1) \sum_{j=1}^n w^{\alpha}_j \ell_j^{1-\alpha}.\]
From \eqref{eq:wl}, we have $(\alpha-1) w^{\alpha}_j \ell_j^{-\alpha}=\sum_{k=j}^n p_{k}$,
which implies
\[(\alpha-1)E(\vec{\ell},\vec{w})=\sum_{j=1}^n \ell_j \sum_{k=j}^n p_{k}=\sum_{j=1}^n p_j \sum_{k=1}^j \ell_k=F(\pi,\vec{\ell}, \vec{p}),\]
which holds by \eqref{ref:relationflowtime}.
\end{proof}

\section{Final remark}
The standard quality measure of the outcome of a game, is the ratio between the social cost of the Nash equilibria and the optimal social cost. Although we have that the ratio is between $1$ and $\alpha+1$ when the job order is fixed, this constant upper bound does not hold when considering the optimal social cost for the best possible order, which might differ from the fixed order in the game. This observation motivates future work: the study of a different game, where the regulator organizes an auction over the rank positions of the schedule.

Finally, we leave open the question about the return of the overcharging to the players for our mechanism, which could improve the ratio between the social cost of the Nash equilibria and the optimal social cost, such as several authors have proposed for VCG mechanisms in the environmental economics setting (see \cite{Varian:94:SolutionExternalities,GrovesLedyard:77:OptimalAllocation,DugganRoberts:02:ImplementingEfficientAllocation,Montero:07:Extension}).\\

\noindent \textbf{Acknowledgements}\\\\
Christoph Dürr and Oscar C. Vásquez were partially supported by grant ANR-11-BS02-0015.
{\L}ukasz Je{\.z} was partially supported by Israeli Centers of Research Excellence (I-CORE)
program, Center No.4/11, Polish National Science Centre (NCN) grant DEC-2013/09/B/ST6/01538, 2014-2017,
and Foundation's for Polish Science (FNP) Start scholarship.

\bibliographystyle{abbrv}
\bibliography{all}
\end{document}